\documentclass[journal]{IEEEtran}
\usepackage{cite}
\usepackage{amsmath,amssymb,amsfonts,amsthm}
\usepackage{algorithmic} 
\usepackage{graphicx}
\usepackage{textcomp}
\usepackage{tikz}
\usepackage{graphicx}
\usepackage{pgfplots}
\usepackage{cuted}
\usepackage{mathrsfs}  
\usepackage[shortlabels]{enumitem}
\usepackage[normalem]{ulem}
\usepackage{bbm}
\usepackage{caption}
\usepackage{subcaption}
\usepackage{comment}
\usepackage{commath}
\usepackage{stfloats}
\usepackage{amsmath}
\usepackage{titlesec}
\usepackage{balance}
\DeclareMathOperator{\Tr}{Tr}

\usetikzlibrary{spy,backgrounds}
\usepackage{lipsum}
\pgfplotsset{width=6.5cm,compat=1.18}
\DeclareRobustCommand{\rchi}{{\mathpalette\irchi\relax}}
\newcommand{\irchi}[2]{\raisebox{\depth}{$#1\chi$}}
\newtheorem{theorem}{Theorem}

\newtheorem{lemma}{Lemma}

\newtheorem{corollary}{Corollary}

\newtheorem{definition}{Definition}

\newtheorem{example}{Example}

\theoremstyle{remark}
\newtheorem{remark}{Remark}

\usepackage{tikz}
\usepackage{caption}
\usepackage{float}
\usetikzlibrary{decorations.pathreplacing}

\newenvironment{psmallmatrix}
  {\left(\begin{smallmatrix}}
  {\end{smallmatrix}\right)}

\newcommand\numeq[1]%
  {\stackrel{\scriptscriptstyle(\mkern-1.5mu#1\mkern-1.5mu)}{=}}
  \newcommand\numl[1]%
  {\stackrel{\scriptscriptstyle(\mkern-1.5mu#1\mkern-1.5mu)}{<}}
\newcommand\numleq[1]%
  {\stackrel{\scriptscriptstyle(\mkern-1.5mu#1\mkern-1.5mu)}{\leq}}
\newcommand\numgeq[1]%
  {\stackrel{\scriptscriptstyle(\mkern-1.5mu#1\mkern-1.5mu)}{\geq}}
\usepackage{mathtools}
\DeclarePairedDelimiter\bra{\langle}{\rvert}
\DeclarePairedDelimiter\ket{\lvert}{\rangle}
\DeclarePairedDelimiterX\braket[2]{\langle}{\rangle}{#1 \delimsize\vert #2}

\usepackage{breqn}

\usepackage{xcolor}
\def\BibTeX{{\rm B\kern-.05em{\sc i\kern-.025em b}\kern-.08em
    T\kern-.1667em\lower.7ex\hbox{E}\kern-.125emX}}
    
    \usepackage{cuted}
\usepackage{lipsum}

\usepackage[normalem]{ulem}

\ifCLASSINFOpdf 
\else  
\fi
\begin{document}

\title{Polar Codes Achieve Classical Capacity for Erasure and Unital Markovian Quantum Channels}
\author{\IEEEauthorblockN{Jaswanthi Mandalapu\IEEEauthorrefmark{1}, Vikesh Siddhu \IEEEauthorrefmark{3}, Krishna Jagannathan\IEEEauthorrefmark{1}\IEEEauthorrefmark{2} \\ \IEEEauthorrefmark{1} ee19d700@smail.iitm.ac.in, \IEEEauthorrefmark{2} krishnaj@ee.iitm.ac.in \\}
\IEEEauthorblockA{\IEEEauthorrefmark{1} Department of Electrical Engineering, IIT Madras  \\
\IEEEauthorrefmark{2} Centre for Quantum Information, Communication and Computing (CQuICC),  IIT Madras} \\
\IEEEauthorrefmark{3} IBM Quantum, IBM Research India
}
\vspace{-4in}

\maketitle

\begin{abstract}

We consider classical-quantum (cq-)channels with  memory, and establish that Arıkan-constructed polar codes achieve the classical capacity for two key noise models, namely for (i) qubit erasures and (ii) unital qubit noise with channel state information at the receiver. The memory in the channel is assumed to be governed by a discrete-time, countable-state, aperiodic, irreducible, and positive recurrent Markov process. We establish this result by leveraging classical polar coding guarantees established for finite-state, aperiodic, and irreducible Markov processes [FAIM] in \cite{csacsouglu2019polar}, alongside the finding that \emph{no entanglement is required} to achieve the capacity of Markovian unital and erasure quantum channels when transmitting classical information \cite{siddhu2024unital, mandayam2020classical}. More broadly, our work illustrates that for cq-channels with memory, where an optimal coding strategy is \emph{essentially classical,} polar codes can be shown to approach the capacity.
\end{abstract}
\pagenumbering{gobble}

\IEEEpeerreviewmaketitle

\section{Introduction}

Quantum information processing provides a novel way to carry out computation, communication, and cryptography using the principles of quantum mechanics.
Quantum channels are central to quantum information processing, they transmit qubits -- the quantum counterparts of classical bits, often entangled -- enabling non-classical phenomena.
However, the practical implementation of quantum communication poses challenges, primarily due to quantum noise (arising from qubit decoherence and thermal effects), which can compromise the integrity of the transmitted information.

The simplest type of information that can be sent across a quantum channel is classical information. The Holevo-Schumacher-Westmoreland (HSW) theorem characterizes the classical capacity of a quantum channel---the maximum rate at which classical information can be reliably transmitted over it. The theorem also establishes the existence of codes that approach this classical capacity. Recent advances in coding theory have begun to focus on explicit coding schemes to achieve this capacity \cite{wilde2011classical,wilde2013polar,wilde2012quantum}.
However, many of these constructions rely on idealized assumptions, such as the independent behavior of the quantum channel across successive uses. In practice, quantum channels often exhibit memory effects, where noise correlations between successive uses complicate the development and analysis of effective coding strategies.

This paper examines quantum channels governed by a discrete-time, countable state, irreducible, aperiodic, and positive recurrent Markovian noise process. Specifically, focusing on erasure and unital noise models, our work aims to design capacity-approaching codes for \emph{classical} information transmission over Markovian quantum channels. In particular, by exploiting the single-letter capacity expressions derived in \cite{mandayam2020classical, siddhu2024unital}, we extend and study the classical capacity achievability using polar codes, originally introduced in \cite{arikan2009channel}, within the context of quantum Markovian channels.

\subsection{Related Work and Our Contributions}
Arıkan proposed low-complexity coding schemes known as polar codes in \cite{arikan2009channel}, which are capacity-approaching for any binary symmetric memoryless classical channel. Underpinning polar codes is the phenomenon of channel polarization, an effect wherein one can analyze a set
of \emph{synthesized channels}, by ``channel combining” and ``channel splitting.” In this process, a fraction of the synthesized channels become `good' for data transmission, while the other channels are `bad' and useless for data transmission. Further, the fraction of good channels converges to channel capacity for larger block lengths, enabling efficient data transmission. Encouraged by these codes in the classical regime, and the well-established HSW coding theorem for quantum memoryless channels, Wilde and Guha~\cite{wilde} extended Arıkan's ideas to the latter setting. Specifically, they proved an explicit polarization of quantum channels and proposed an encoder and quantum successive cancellation decoder (QSCD) for efficient classical information transmission. 

In parallel, recent papers \cite{sasoglu2011polar,csacsouglu2019polar,shuval2018fast,7360760} have also considered the polarization properties of classical channels with memory, and proved that polar codes can be used directly to approach the capacity for fast-mixing classical channels. Motivated by these developments, in this work, we focus on quantum memory channels modeled by Markov processes and propose classical capacity-achieving polar codes for erasure and unital noise models. In summary, our main contributions are as follows:
\begin{enumerate}
    \item \emph{Equivalence to Induced Classical Channel:} For simple noise models like erasure and unital channels (with known channel state information), prior works \cite{mandayam2020classical,siddhu2024unital} showed that encoding of classical bits into \emph{untangled, orthogonal
    } quantum states achieves the classical-capacity. Consequently, we demonstrate that the classical coding achieves the classical capacity for such channels.
    \item \emph{Polar Coding Guarantees for Markovian Channels:} Next, by approximating countable-state Markov processes with finite-state counterparts, and leveraging Portmanteau theorem \cite[Sec.~2.5]{van2000asymptotic}, we extend the results of \cite{csacsouglu2019polar} to countable state space channels. This step completes the achievability proof of classical capacity using Arıkan polar codes for erasure and unital Markovian quantum channels.
\end{enumerate}
The rest of the paper is organized as follows: Sec.~\ref{sec-2} introduces the erasure and unital Markovian channel models and summarizes existing results on their classical capacities. Sec.~\ref{sec:polar-FAIM} reviews classical polar coding schemes and known results for finite-state channels. Finally, building on these foundations, Sec.~\ref{sec-4} presents our main results and proofs.

\section{The Erasure and Unital Markovian CQ-Channels and their Capacities} \label{sec-2}
\subsection{The Markovian Classical-Quantum (CQ-)Channel}\label{faim-cq-channel}
In this section, we introduce the framework of a cq-channel governed by a discrete-time, countable-state, aperiodic, irreducible and positive recurrent Markovian noise process. 
Let $\mathbb{N} = \{1,2,\ldots\}$ be the set of natural numbers.
In a Markovian cq-channel, $n \in \mathbb{N}$ input bits $X_1^n \in \mathcal{X}^n$, $\mathcal{X} = \{0,1\}$, are encoded into a \emph{noiseless,} possibly entangled, quantum state $\sigma_X^{(n)} $ over an $n$-partite input Hilbert space $\mathcal{H}_A^{(n)}$.
Transmitting this state over an $n$-use quantum memory channel
$\mathcal{N}^{(n)}$ (defined below) results in an output state $\rho_X^{(n)}$ over the output quantum systems $B_1^n,$ with the Hilbert space $\mathcal{H}_B^{(n)}$, i.e.
\begin{equation}
    \rho_X^{(n)} =  \mathcal{N}^{(n)}(\sigma_X^{(n)}).
\end{equation}
From this received quantum state $\rho_X^{(n)},$ we obtain an estimate of
transmitted input $\hat{X}_1^n.$ Specifically, we define our classical-quantum (cq-)channel as 
\begin{align}\label{defn:cq-channel}
    W^{(n)}:X_1^n \to \rho_X^{(1,n)},
\end{align}
where $W^{(n)} = \mathcal{N}^{(n)} \circ \mathcal{E}^{(n)}$ such that
$\mathcal{E}^{(n)} : X_1^n \to \sigma_X^{(n)}$ is the noiseless encoding
operation defined above. Further, we use $\mathcal{B}(\mathcal{H}_A^{(n)}),$ and $
\mathcal{B}(\mathcal{H}_B^{(n)})$ to indicate the sets of all density operators on
the Hilbert spaces $\mathcal{H}_A^{(n)}$ and $\mathcal{H}_B^{(n)}$,
respectively.

We now provide the formal definition of a Markovian quantum channel
$\mathcal{N}^{(n)}$ and subsequently recall its classical capacity from \cite{datta2009classical,hayashi2003general}.
\begin{definition}
Let $\mathcal{K}$ denote the countable state space of a underlying Markov chain governing the noise process $\{K_i, i \in \{1,\ldots,n\}\}$. Further, let $\mathbf{K} = \{k_{ss'}\}_{s,s' \in \mathcal{K}}$ denote its transition probability matrix,
and let $\{\pi_{s}\}_{s \in \mathcal{K}}$ represent the stationary distribution, i.e., $\pi_{s'} = \sum_{s \in \mathcal{K}} \pi_s k_{ss'}, ~\forall s'
\in \mathcal{K}.$ Then, an $n-$use quantum Markovian channel
$\mathcal{N}^{(n)}$ can be defined as 
\begin{equation}\label{FAIM-Channel}
\begin{split}
&\mathcal{N}^{(n)}\left(\sigma^{(n)}\right) \\
    & = \sum_{s_1,\ldots,s_n \in \mathcal{K}} \pi_{s_1} k_{s_1 s_2} \ldots k_{s_{n-1} s_n} \left(\mathcal{N}_{s_1} \otimes \ldots \otimes \mathcal{N}_{s_n}\right) \left(\sigma^{(n)}\right),
\end{split}
\end{equation}
where $\mathcal{N}_s:\mathcal{B}(\mathcal{H}_A) \to \mathcal{B}(\mathcal{H}_A)$
is a completely positive trace-preserving (CPTP) map for each $s \in
\mathcal{K}.$
\end{definition}
With out loss of generality, we assume $\mathcal{K} = \mathbb{Z}^+$, and in \eqref{FAIM-Channel} that the Markov process is started (and
hence remains) in the stationary distribution $\{\pi_{i}\}_{i \in \mathcal{K}}$
throughout the transmission process.

\subsection{The Classical Capacity of a general CQ-Channel}
 An error correcting code for an arbitrary cq-channel can be defined as follows; see \cite{hayashi2003general} for more information.
\begin{definition}
A code $\mathcal{C}^{(n)}$ of size $N_n$ is given by a sequence $\{\sigma_i^{(n)}, M_i^{(n)}\}_{i=1}^{N_n}$, where $\sigma_i^{(n)}
    \in \mathcal{B}(\mathcal{H}_A^{(n)})$, and $M_i^{(n)}$ are positive semi-definite
    operators on $\mathcal{H}_B^{(n)}$ such that ${\sum_{i=1}^{N_n}
    M_i^{(n)} \leq \mathbf{1}^{(n)}}$; here, $\mathbf{1}^{(n)}$ denotes an
    identity operator in $\mathcal{B}(\mathcal{H}_B^{(n)})$. 
\end{definition}

Let $M_0^{(n)} = \mathbf{1}^{(n)} - \sum_{i=1}^{N_n} M_i^{(n)}$, performing 
    Positive Operator Valued Measurement (POVM) 
    $\{M_i^{(n)}\}$ yields an output $i$, where
    $i \neq 0$ corresponds to
    $\sigma_i^{(n)}$ being transmitted through $\mathcal{N}^{(n)}$, and $i = 0$ corresponds to failure to decode.
    
\begin{definition}
    The average probability of error for the code $\mathcal{C}^{(n)}$ used across $\mathcal{N}^{(n)}$ is given by
    \begin{align*}
        P_e(\mathcal{C}^{(n)}) \coloneqq \frac{1}{N_n} \sum_{i=1}^{N_n} \left( 1 - \Tr (\mathcal{N}^{(n)}(\sigma_i^{(n)}) M_i^{(n)}\right).
    \end{align*}
\end{definition}

\begin{definition}
    A rate $R$ is said to be achievable, if there exists an $n_0 \in \mathbf{N}$ such that for all $n \geq n_0,$ there exists a sequence of codes $\{\mathcal{C}^{(n)}\}_{n=1}^{\infty},$ of sizes $N_n \geq 2^{nR},$ for which $P_e(\mathcal{C}^{(n)}) \to 0$ as $n \to \infty.$ Further, the classical capacity of channel $\mathcal{N}$ is defined as 
    $$
        C(\mathcal{N}) \coloneqq \sup R,
    $$
    where $R$ is an achievable rate.
\end{definition}
 For any arbitrary cq-channel, the classical capacity is evaluated using the quantum analogue of classical information-spectrum analysis introduced in \cite{hayashi2003general}. Specifically, let $\vec{\mathcal{N}}$ denote the sequence of quantum channels $\{\mathcal{N}^{(n)}\}_{n=1}^{\infty}$; $\vec{P}$ denote the totality of sequences $\{P^n(X^n)\}_{n=1}^{\infty}$ of probability distributions (with finite support) over input sequences $X_1^n$; $\vec{\sigma}$ denote the sequences of states $\{\sigma_X^{(n)}\}$ corresponding to the encoding $\mathcal{E}':X_1^n \to \sigma_X^{(n)}.$ Then,

\begin{theorem}\cite{hayashi2003general}
    The classical capacity of any arbitrary quantum channel is given as 
    \begin{align}\label{eq:capacity-general}
        C(\mathcal{N}) = \sup_{\vec{P},\vec{\rho}} \underline{{I}}(\{\vec{P},\vec{\rho}\}, \vec{N}),
    \end{align}
    where $\underline{{I}}(\{\vec{P},\vec{\rho}\}, \vec{N})$ is the quantum analog (see \cite[eq.~6]{hayashi2003general}) of inf-information rate defined in \cite{verdu1994general}.
\end{theorem}

\subsection{Noise Model}\label{sec:noise}
Under the above framework, we consider two noise models: (i) Erasure noise and (ii) Unital noise. For erasure noise we let each $\mathcal{N}_{K}$ in \eqref{FAIM-Channel} be an erasure channel,
\begin{align*}
    \mathcal{N}_K(\sigma) = (1-p(K))\sigma + p(K) \ket{e}\bra{e},
\end{align*}
where $\ket{e}\bra{e}$ is an erasure operator orthogonal to the input Hilbert space $\mathcal{H}_A$, and erasure probability $p(K)$ is a function ${p: \mathcal{K} \to [0,1]}$ mapping states in $\mathcal{K}$ to values in $[0,1]$. Next, we define a general unital qubit noise model by letting each $\mathcal{N}_{K}$ in~\eqref{FAIM-Channel} be a unital qubit channel, i.e, each $\mathcal{N}_{K}$ satisfies
\begin{align*}
    \mathcal{N}_K(\mathbb{I}_A) = \mathbb{I}_B,
\end{align*}
where $\mathbb{I}_A$ and $\mathbb{I}_B$ are the Identity operators on the input and output Hilbert spaces, $\mathcal{H}_A$ and $\mathcal{H}_B$, respectively.

\begin{example}
    \normalfont As an example of Markovian erasure noise, we introduce the `Gilbert-Elliott quantum erasure channel'~(GE-QEC). Note that this is an analogous model of the classical Gilbert-Elliott channel model \cite{mushkin1989capacity} to the quantum domain, where the erasure noise depends on the state of an underlying Markov process. In GE-QEC, we assume that the probability of erasure of a particular qubit $i$ depends on its noise state $K_i,$ which is denoted by $p(K_i).$ We further assume that the noise state $\{K_i, i \in \mathbb{N}\}$ evolves according to a discrete-time, two-state, aperiodic, and irreducible Markov chain as depicted in Fig.~\ref{fig:GGE_QEC}.
    \begin{figure}[t!]
     \centering
     \begin{tikzpicture}[node distance = 2cm]
        \node [circle,draw] (zero) {0};
        \node [circle,draw] (one) [right of=zero] {1};
        \path [->] (zero) edge [bend left] node [above] {$k_{01}$} (one) ;
        \path [->] (one) edge [bend left] node [below] {$k_{10}$} (zero) ;
        \path [->] (zero) edge [loop left] node [below] {$k_{00}$} (zero) ;
        \path [->] (one) edge [loop right] node [below] {$k_{11}$} (one) ;        
     \end{tikzpicture}
     \caption{The state transition diagram of a GE-QEC, where $k_{ij}$ represents the transition probability from state $i$ to state $j$.}
     \label{fig:GGE_QEC}
 \end{figure}
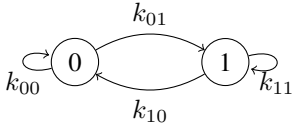
 \end{example}

 \begin{example}
\normalfont As an example of Markovian unital noise, we recall the qubit depolarizing quantum \emph{queue-channel} introduced in \cite{mandayam2020classical}, where the noise depends on the queue lengths seen at each arrival into the queue. In particular, we assume qubits are processed in a single-server queue according to First Come, First Serve (FCFS) discipline; the qubit arrival process is a Poisson process of fixed rate of $\lambda$ qubits per second, and the service times are independent and exponentially distributed with a mean of $\mu$. We assume the noise in the system arises as the qubits experience larger queuing delays in the queue. Specifically, if $Q_i$ indicates the queue-length seen at the $i^{th}$ arrival into the queue, then for the depolarizing noise, we assume each $\mathcal{N}_{Q_i}$ in~\eqref{FAIM-Channel} is defined as
\begin{align*}
    \mathcal{N}_{Q_i}(\sigma_i) = (1-p(Q_i))\sigma_i + p(Q_i) \frac{\mathbb{I}_{B}}{2},
\end{align*}
where $p(Q_i)$ is the noise probability that depends on the queue-length $Q_i.$ Furthermore, it is known that the queue-lengths seen at each arrival in FCFS service discipline are governed by Lindley's recursion, and hence form a discrete-time, countable-state, irreducible Markov process; see \cite{mandayam2020classical,chatterjee2017capacity} for more details.
 \end{example}

Note that both of these examples follow the Markovian cq-channel framework introduced earlier. We now proceed to recall existing results on the classical capacities of erasure and unital Markovian noise models derived recently in \cite{mandayam2020classical} and \cite{siddhu2024unital}. First, for the erasure noise, let ${\varepsilon_{\pi} =
\mathbb{E}_{\pi}[p_{K_i}]=\sum_{K_i\in\mathcal K}\pi_{K_i}p_{K_i}}$ denotes the stationary erasure probability. Then,
\begin{theorem}\label{thm-erasure}\cite{mandayam2020classical}
    The classical capacity of an `erasure' Markovian cq-channel is given by
    $
        {C(\mathcal{N}) = 1 - \varepsilon_{\pi} }
    $
     bits per channel use. Further, encoding of the classical bits into orthogonal and product quantum states
     achieve the capacity.
\end{theorem}
\noindent

Next, let $\{q_j\}$ be a probability distribution and $\{\rho_j\}$ be quantum states in $\mathcal{B}(\mathcal{H}_A)$. Then, Holevo information of the ensemble, $\{p_j, \rho_j\}$, $$\chi (\{p_j, \rho_j\}) = S\big( \mathbb{E}_{q}[\rho] \big) - \mathbb{E}_{q}[S(\rho_j)],$$ where $S(\rho) = - \Tr (\rho \log \rho)$ is the von-Neumann entropy. Further, the Holevo information 
$\chi(\mathcal{M})$ of any channel $\mathcal{M}:\mathcal{B}(\mathcal{H_A}) \mapsto \mathcal{B}(\mathcal{H}_B)$ is the maximum over ensembles $\{p_i, \mathcal{N}(\rho_i)\}$. Following this, for the unital noise, we have
\begin{theorem}\label{thm-unital} \cite{siddhu2024unital}
 For the `unital' Markovian cq-channel where the receiver knows the channel state information, the classical capacity is given by
    $
        {C({\mathcal{N}}) = \mathbb{E}_{\pi}\left[\rchi(\mathcal{N}_K)\right]}
    $
    bits per channel use, where $\rchi(.)$ is the Holevo quantity defined above. Furthermore, encoding of the classical bits into orthogonal and product quantum states
    achieve the capacity under two cases: (i) when the encoder knows the channel state information, or (ii) when the unital qubit Markovian channel is Pauli-ordered (see \cite[defn.~4]{siddhu2024unital}).
\end{theorem}
\begin{corollary}
     For the `depolarizing' Markovian cq-channel with known channel state information at the receiver, the classical capacity is given by $
        C({\mathcal{N}}) = \mathbb{E}_{\pi}\left[1 - h\left( \frac{p(K)}{2}\right)\right] ~~~\text{bits per channel use},
    $   
    where $h(.)$ denotes the binary entropy function. Further, encoding of the classical bits into orthogonal and product quantum states achieve the capacity.\end{corollary}
\begin{remark}
    When the receiver does not have knowledge about the underlying states of a depolarizing channel the classical capacity is not known precisely. Specifically, the capacity is known to be between $1 - h\left( \frac{\mathbb{E}_{\pi}[p(K)]}{2}\right)$ and $\mathbb{E}_{\pi}\left[1 - h\left( \frac{p(K)}{2}\right)\right]$ bits per channel use. These upper and lower bounds are separated by a \emph{Jensen gap}.
    \end{remark}
Next, prior to fulfilling our aim of presenting explicit capacity-approaching polar codes for erasure and unital qubit Markovian cq-channels, we review classical polar coding and its achievable guarantees for classical Markovian channels in the next section.

\section{Classical Polar Coding and its \\ Achievable Guarantees}\label{sec:polar-FAIM}
%
Let $(X_i,Y_i)$, $i \in \mathbb{N}$ be a classical stationary and ergodic process, with $X_i$'s, and $Y_i$'s indicating the inputs and outputs of a traditional classical noisy channel denoted by $(X^{(n)},Y^{(n)},P_{Y^{(n)}|X^{(n)}})$, respectively.
We assume that $X_i$'s are binary and $Y_i \in \mathcal{Y}$, where $\mathcal{Y}$ is a finite alphabet; further, the outputs $Y^{(n)}$ are governed by a Finite-state, Aperiodic, and Irreducible Markov (FAIM) process. 
We refer to $(X^{(n)},Y^{(n)},P_{Y^{(n)}|X^{(n)}})$ as a `FAIM classical channel' --- an example is the Gilbert-Elliot channel introduced in \cite{mushkin1989capacity}. 

 Let the polar transformed channel input be given as\cite{arikan2009channel}: 
 \begin{align}\label{eq:polartransform}
     U_1^N = X_1^N F_N G_N,
 \end{align}
where $N=2^n$ is the block length for some $n \in \mathbb{N}$, $F_N$ is the $N \times N$ bit reversal matrix, and $G_N$ is the $n^{th}$ Kronecker power of the matrix $\begin{psmallmatrix} 1 & 0\\1 & 1\end{psmallmatrix}$. 
Define the conditional entropy rate as
\begin{equation}\label{entropy-rate}
\begin{split}
    \mathcal{H}_{X|Y} &= \lim_{N \to \infty} \frac{1}{N} H(X_1^N|Y_1^N). \\
\end{split}
\end{equation}
\begin{lemma}\label{thm:slow-po-sasoglu}\cite{sasoglu}
    The above defined FAIM channel satisfies
    \begin{align*}
         \lim_{N \to \infty} \frac{1}{N} |\{i : \mathit{I}(U_i; U_1^{i-1},Y_1^N) < \gamma\}| &= 1 - \mathcal{H}_{X|Y},\\
     \lim_{N \to \infty} \frac{1}{N} |\{i : \mathit{I}(U_i; U_1^{i-1},Y_1^N) > 1 - \gamma\}| &=  \mathcal{H}_{X|Y},
    \end{align*}
    for all $\gamma > 0.$
Additionally, we also observe
     \begin{align*}
         \lim_{N \to \infty} \frac{1}{N} |\{i : \mathit{Z}(U_i| U_1^{i-1},Y_1^N) < 2^{-N^{\beta}}\}| = 1 - \mathcal{H}_{X|Y},
     \end{align*}
     for $\beta < 1/2.$
\end{lemma}

Note that we have used standard notation $\mathit{I}(\cdot ; \cdot)$ and $\mathit{Z}(\cdot|\cdot)$ to denote the classical mutual information ~(with uniform channel input) and Bhattacharyya parameter, respectively~(as in \cite{sasoglu}).
Following these polarization results, it is also straightforward to show the following theorem for FAIM channels.
\begin{theorem}
    \label{FAIM-channels}
   For a FAIM classical channel, standard Arıkan's polar transform from \cite{arikan2009channel} can achieve rates $ {\mathrm{R} < 1 - \mathcal{H}_{X|Y}}$ bits per channel use, for sufficiently large enough block length $N$. Further, the block error probability decays at a rate of $O(2^{-N^{\beta}}), \beta < 1/2$ under successive cancellation trellis decoding. 
\end{theorem}

Using the results from the preceding sections as our building blocks, we now proceed to prove our main results. Specifically, we prove that classical polar encoding and decoding are sufficient to achieve the capacity for erasure and unital Markovian classical-quantum channels with a countable state space.

\section{Main Results and Proofs}\label{sec-4}
\begin{theorem}\label{thm-erasure-capacity}
    The capacity for an erasure Markovian cq-channel, stated in
    Theorem~\ref{thm-erasure}, can be achieved using standard Arıkan polar transformation \eqref{eq:polartransform} applied to classical bits before encoding them into quantum states along with product decoding of quantum states into classical bits at the receiver. Successive cancellation trellis decoding of these classical bits achieves block error probability of $O(2^{-N^{\beta}}), \beta < 1/2$ for sufficiently large enough block length $N.$
\end{theorem}

\noindent
\textit{Proof Outline:} We prove the above theorem using the following four-step approach.
\begin{itemize}
    \item \textit{Step-1:} First, we note that, as a consequence of Theorem~\ref{thm-erasure}, encoding of classical bits into orthogonal quantum states and product decoding suffices to achieve the classical capacity. Thus, we focus on approaching the capacity of the classical channel induced by this encoding-decoding\footnote{The induced classical channel of a cq-channel refers to the channel between the classical input bits $X_1^N$ and the measured quantum bits, denoted by $Y_1^N$.} using classical encoding and decoding, specifically with Arıkan polar transformation \cite{arikan2009channel}.
    \item \textit{Step-2:} Second, to help analyze polar coding of the induced channel we demonstrate that discrete-time, countable-state, aperiodic, and irreducible Markov process satisfying certain conditions can be truncated and approximated by a discrete-time, finite-state, aperiodic, and irreducible Markov process (FAIM), such that their steady state probabilities converge.
    \item \textit{Step-3:} Third, following the results from Sec.~\ref{sec:polar-FAIM}, we prove the capacity-achievability for the approximated FAIM process with polar codes.
    \item \textit{Step-4:} Finally, leveraging the Portmanteau Theorem, we establish the capacity achievability for the erasure Markovian cq-channel using classical polar codes.
\end{itemize}

We remark that the same result holds true for unital Markovian cq-channels when the receiver has knowledge of the channel states. In this case, the capacity achievability for unital Markovian cq-channels follows from analogous arguments, along with the results as stated in Theorem~\ref{thm-unital}. Specifically, we also have

\begin{theorem}
    For a unital Markovian cq-channel with `known channel states at the receiver', classical polar encoding and decoding  achieves its classical capacity when the encoder also knows about the channel state information or if the unital Markovian cq-channel is Pauli ordered; further, the block error probability decays at the rate of $O(2^{-N^{\beta}}), \beta < 1/2$ for sufficiently large block length $N.$
\end{theorem}
\subsection{Proof of Theorem~\ref{thm-erasure-capacity}}\label{proof-thm-capacity}
Note that \emph{Step-1} directly follows from the capacity results as stated in Theorems~\ref{thm-erasure}~and~\ref{thm-unital} from Sec.~\ref{sec:noise}. We now proceed to the prove the later steps.
\noindent
\emph{Step-2:} Recall, $\mathbf{K} = \{k_{ss'}\}_{s,s' \in \mathcal{K}}$ is the transition probability matrix of the discrete-time, countable-state, aperiodic, irreducible, and positive recurrent Markov process with state space denoted by $\mathcal{K}$. Since the Markov process is irreducible and positive recurrent, it has a unique stationary distribution, which we denoted by $\{\pi_s\}$, satisfying $\pi_s > 0,~\forall~s, \sum_{s \in \mathcal{K}} \pi_s = 1$. Define $\mathbf{K}^{(l)}$ and $\mathcal{K}_l = \{0,1,2,\ldots,l\}$ as the north-west corner truncation\footnote{The north-west corner truncation is simply considering the upper-left part of the matrix corresponding to the truncation parameter $l.$} of the matrix $\mathbf{K},$ and the underlying state-space, respectively. Further, let $\{\tilde{\pi}^{(l)}_s\}$ be the stationary distribution obtained from $(l+1) \times (l+1)$ stochastic matrix $\tilde{\mathbf{K}}^{(l)}$ over $\mathcal{K}_l$, where it holds $\tilde{\mathbf{K}}^{(l)} \geq \mathbf{K}^{(l)},$ element-wise. 

In addition, let $\mathcal{C} \subset \mathcal{K}$ be a finite set, and $n_{\mathcal{C}} = \min_{n \in \mathcal{C}} n$ 
Then, we have the following lemma.
\begin{lemma}\label{lemma:finite-approximation}
 For a discrete-time, countable-state, irreducible, and positive recurrent Markov chain with state space $\mathcal{K}$, the sequence of stationary distributions $\{\{\tilde{\pi}^{(l)}\}\}$ converges in $L^{1}$ norm to $\{\pi_s\}$ if there exists a finite set $\mathcal{C}$ such that the matrices $\mathbf{K}$ and $\{\tilde{\mathbf{K}}^{(l)}, l \geq n_{\mathcal{C}} \}$ uniformly satisfy:
 \begin{enumerate}[a1)]
     \item  For any $s \in \mathcal{C}$ and any $B \subseteq \mathcal{K}, p(s,B) \geq p \varphi(B),$ where $p \in (0,1)$ and $c < \infty$ are some constants, and $\varphi$ is some probability measure on $\mathcal{K}.$
     \end{enumerate}
Furthermore, $\mathbf{K}$ satisfies
\begin{enumerate}[a2)]
    \item There exists a non-negative and non-decreasing function $V:\mathcal{K} \to [0,\infty)$ and a positive constant $b < \infty$ such that 
    $$
        \mathbf{K}V \leq V - 1 + b\mathbb{I}_{\mathcal{C}}.
    $$
\end{enumerate}
\end{lemma}
Note that the Markov chains satisfying $(a1)$ and $(a2)$ are said to be \emph{positive Harris-recurrent.} Refer to \cite{liu2010augmented,gibson1987augmented,wolf1980approximation,tweedie1998truncation,infanger2022convergence,seneta1980computing}, \cite[Chapter~11]{meyn2012markov} for more details.
\begin{proof}
    See proof of \cite[Theorem~2.3]{liu2010augmented}.
\end{proof}

\begin{remark}
    Some simple structures of Markov chains satisfying the above criterion are the following:
    \begin{itemize}
        \item The transition matrix $\mathbf{K}$ is Upper-Hessenberg, i.e., $k_{ss'} = 0,$ if $s > s'+1,$  or Lower-Hessenberg, i.e.,  $k_{ss'} = 0,$ if $s' > s+1.$
        \item $\mathbf{K}$ is Markov matrix, i.e., if there exists a state $j$ such that $k_{sj} > \delta > 0$ for all $s.$
        \item $\mathbf{K}$ is dominated by stochastically monotone matrix $\mathbf{M}$, i.e., for any $s,j$ and $s < j$ and for every $m$,
        \begin{align*}
            \sum_{r > m} \mathbf{M}_{sr} \leq \sum_{r > m} \mathbf{M}_{jr},
        \end{align*}
        and for $s,m \in \mathcal{K}$
        \begin{align*}
            \sum_{r > m} \mathbf{M}_{sr} \geq \sum_{r > m} \mathbf{K}_{sr}.
        \end{align*}       
        \end{itemize}
Note that most Markov chains encountered in practical communication systems satisfy one of these conditions.
\end{remark}

\noindent
\emph{Step-3:} Consider an erasure cq-channel with noise governed by a discrete-time, finite-state, aperiodic, and irreducible Markov process $\{K_i \in \mathcal{K}_l, i \in \{1,\ldots,n\}\}$; let $\tilde{\mathbf{K}}^{(l)}$ be its transition probability matrix, and $\tilde{\pi}^{(l)}$ denotes its steady-state probabilities. Then, 
\begin{lemma}\label{lemma: pola-finite-approximation}
    For any rate $R < 1 - \mathbb{E}_{\tilde{\pi}^{(l)}} [1 - p(K)]$, the block error rate decays at rate of $O(2^{-\sqrt{N}})$ for sufficiently large enough block length $N,$ under standard Arıkan polar transformation and successive cancellation trellis decoding. 
\end{lemma}
\begin{proof}
    The proof of this theorem directly follows from the results stated in Theorems~\ref{thm-erasure}~and~\ref{FAIM-channels}. See \cite{csacsouglu2019polar} for more details.
\end{proof}

\noindent
\emph{Step-4:} We now prove the achievability of capacity using polar codes for countable-state Markovian channels. Let $U_1^N$ indicate the polar transformed inputs, and $Y_1^N$ indicate the classical outputs after performing an appropriate product decoding along with measurement, of a given cq-channel. Consider a finite-length approximation of Markov process $\mathbf{K}$, denoted by $\mathbf{K}^{(l)}$, constructed as stated in Lemma~\ref{lemma:finite-approximation}. Then, following Lemmas~\ref{lemma: pola-finite-approximation}~and~\ref{thm:slow-po-sasoglu}, we have

\begin{equation}\label{eq:faim}
\begin{split}
    \lim_{N \to \infty} \frac{1}{N} \lvert \{i: I(U_i: U_1^{i-1}, Y_1^N) > 1 - \gamma\}| &= 1 - \epsilon_{\tilde{\pi}^{(l)}}, \\
    \lim_{N \to \infty} \frac{1}{N} \lvert \{i: Z(U_i| U_1^{i-1}, Y_1^N) < 2^{-N^{\beta}}\}| &= 1 - \epsilon_{\tilde{\pi}^{(l)}}, 
    \end{split}
\end{equation}
for all $0 < \gamma <1$ and $\beta < 1/2,$ where $\epsilon_{\tilde{\pi}^{(l)}} = \mathbb{E}_{\tilde{\pi}^{(l)}}[p(K)]$. Next, as $l \to \infty,$ from Lemma~\ref{lemma:finite-approximation}, we have $\tilde{\pi}^{(l)} \xrightarrow[]{L^{1}} \pi $. This implies, $\tilde{\pi}^{(l)} \xrightarrow[]{\mathcal{D}} \pi.$ Further, since $p : \mathbb{N} \to [0,1]$ is a bounded and continuous function, from Portmanteau Theorem \cite[Theorem~1.3.4]{wellner2013weak}, we have $\lim_{l \to \infty} \mathbb{E}_{\tilde{\pi}^{(l)}}[p(K)] = \mathbb{E}_{\pi}[p(K)]$. Therefore, from \eqref{eq:faim}, as $l \to \infty,$ we have
\begin{equation}\label{eq:countable}
\begin{split}
    \lim_{l \to \infty} \lim_{N \to \infty} \frac{1}{N} \lvert \{i: I(U_i: U_1^{i-1}, Y_1^N) > 1 - \gamma\}| &= 1 - \epsilon_{\pi}, \\
    \lim_{l \to \infty} \lim_{N \to \infty} \frac{1}{N} \lvert \{i: Z(U_i| U_1^{i-1}, Y_1^N) < 2^{-N^{\beta}}\}| &= 1 - \epsilon_{\pi}, 
    \end{split}
\end{equation}
for all $0 < \gamma <1$ and $\beta < 1/2.$ This concludes Theorem~\ref{thm-erasure-capacity} as Bhattacharyya parameter is known to upper bound the error probability. \qed
\begin{remark}
    Note that the coding scheme described above achieves the capacity of the induced classical channel.
    However, as shown in Theorems~\ref{thm-erasure} and~\ref{thm-unital}, for binary input stationary erasure and unital cq-channels, the classical capacity defined in \eqref{eq:capacity-general} coincides with the capacity of the induced classical channel, which is given by $1 - \mathcal{H}_{X|Y}$.
\end{remark}
\begin{remark}
    For practical decoding of Markovian channels with a countable state space, we propose the following approach: Given a Markovian channel with a transition matrix for the noise states, denoted by $\mathcal{K}$, find a finite approximation $\mathcal{K}_l$ such that the stationary distributions $\pi$ and $\pi^{(l)}$ converge. This approximation can be constructed following the method outlined in Sec. \ref{proof-thm-capacity}. Next, for any state $k>l$, we assume $p(K)=1$. Following this, we apply successive cancellation trellis decoder from \cite{wang2015construction} to decode the information bits from obtained finite-state Markovian channel. Additionally, we assert that the difference in rates can be determined by leveraging the upper bounds on the difference between the stationary distributions of the truncated and untruncated chains, as provided in \cite{tweedie1998truncation}.
\end{remark}

\section*{Conclusion}
We considered a quantum channel where noise is modeled by a discrete-time, countable-state, irreducible, aperiodic, and positive recurrent Markov process. Focusing on erasure and unital noise models (with known channel state information in the latter case), we demonstrated that classical polar codes can achieve the classical capacity of such Markovian channels. Specifically, our proof followed a two-step approach: (i) following prior results on erasure and unital quantum channels, we established that classical coding suffices to achieve the classical capacity, and (ii) we developed a finite-state approximation for the countable-state Markov channels under the given conditions. By combining this approximation with existing capacity achievability results for classical polar codes on finite-state channels, we concluded the proof of classical capacity achievability for Markovian erasure and unital quantum channels.

An intriguing direction for future research is the development of classical capacity-achieving codes for Markovian quantum channels with `non-unital' noise models.
\section*{Acknowledgements}
The authors thank Dr. Avhishek Chatterjee for his helpful suggestions.  V.S. is supported by the U.S. Department of
Energy, Office of Science, National Quantum Information Science Research Centers, Co-design Center for
Quantum Advantage (C2QA) contract (DE- SC0012704)
\balance
\bibliography{References} 

\begin{thebibliography}{10}

\bibitem{csacsouglu2019polar}
E.~{\c{S}}a{\c{s}}o{\u{g}}lu and I.~Tal, ``Polar coding for processes with memory,'' {\em IEEE Transactions on Information Theory}, vol.~65, no.~4, pp.~1994--2003, 2019.

\bibitem{siddhu2024unital}
V.~Siddhu, A.~Chatterjee, K.~Jagannanthan, P.~Mandayam, and S.~Tayur, ``Unital qubit queue-channels: classical capacity and product decoding,'' {\em IEEE Transactions on Quantum Engineering}, 2024.

\bibitem{mandayam2020classical}
P.~Mandayam, K.~Jagannathan, and A.~Chatterjee, ``The classical capacity of additive quantum queue-channels,'' {\em IEEE Journal on Selected Areas in Information Theory}, vol.~1, no.~2, pp.~432--444, 2020.

\bibitem{wilde2011classical}
M.~M. Wilde, ``From classical to quantum shannon theory,'' {\em arXiv preprint arXiv:1106.1445}, 2011.

\bibitem{wilde2013polar}
M.~M. Wilde and S.~Guha, ``Polar codes for degradable quantum channels,'' {\em IEEE Transactions on Information Theory}, vol.~59, no.~7, pp.~4718--4729, 2013.

\bibitem{wilde2012quantum}
M.~M. Wilde and J.~M. Renes, ``Quantum polar codes for arbitrary channels,'' in {\em 2012 IEEE International Symposium on Information Theory Proceedings}, pp.~334--338, IEEE, 2012.

\bibitem{arikan2009channel}
E.~Arikan, ``{Channel polarization: A method for constructing capacity-achieving codes for symmetric binary-input memoryless channels},'' {\em IEEE Transactions on information Theory}, vol.~55, no.~7, pp.~3051--3073, 2009.

\bibitem{wilde}
M.~M. Wilde and S.~Guha, ``Polar codes for classical-quantum channels,'' {\em IEEE Transactions on Information Theory}, vol.~59, no.~2, pp.~1175--1187, 2013.

\bibitem{sasoglu2011polar}
E.~Sasoglu, ``Polar coding theorems for discrete systems,'' tech. rep., EPFL, 2011.

\bibitem{shuval2018fast}
B.~Shuval and I.~Tal, ``Fast polarization for processes with memory,'' {\em IEEE Transactions on Information Theory}, vol.~65, no.~4, pp.~2004--2020, 2018.

\bibitem{7360760}
R.~Wang, J.~Honda, H.~Yamamoto, R.~Liu, and Y.~Hou, ``Construction of polar codes for channels with memory,'' in {\em 2015 IEEE Information Theory Workshop - Fall (ITW)}, pp.~187--191, 2015.

\bibitem{van2000asymptotic}
A.~W. Van~der Vaart, {\em Asymptotic Statistics}, vol.~3.
\newblock Cambridge university press, 2000.

\bibitem{datta2009classical}
N.~Datta and T.~Dorlas, ``Classical capacity of quantum channels with general markovian correlated noise,'' {\em Journal of Statistical Physics}, vol.~134, pp.~1173--1195, 2009.

\bibitem{hayashi2003general}
M.~Hayashi and H.~Nagaoka, ``General formulas for capacity of classical-quantum channels,'' {\em IEEE Transactions on Information Theory}, vol.~49, no.~7, pp.~1753--1768, 2003.

\bibitem{verdu1994general}
S.~Verd{\'u} {\em et~al.}, ``A general formula for channel capacity,'' {\em IEEE Transactions on Information Theory}, vol.~40, no.~4, pp.~1147--1157, 1994.

\bibitem{mushkin1989capacity}
M.~Mushkin and I.~Bar-David, ``Capacity and coding for the gilbert-elliott channels,'' {\em IEEE transactions on Information Theory}, vol.~35, no.~6, pp.~1277--1290, 1989.

\bibitem{chatterjee2017capacity}
A.~Chatterjee, D.~Seo, and L.~R. Varshney, ``Capacity of systems with queue-length dependent service quality,'' {\em IEEE Transactions on Information Theory}, vol.~63, no.~6, pp.~3950--3963, 2017.

\bibitem{sasoglu}
E.~Şaşoğlu and I.~Tal, ``{Polar coding for processes with memory},'' {\em IEEE Transactions on Information Theory}, vol.~65, no.~4, pp.~1994--2003, 2019.

\bibitem{liu2010augmented}
Y.~Liu, ``Augmented truncation approximations of discrete-time markov chains,'' {\em Operations research letters}, vol.~38, no.~3, pp.~218--222, 2010.

\bibitem{gibson1987augmented}
D.~Gibson and E.~Seneta, ``Augmented truncations of infinite stochastic matrices,'' {\em Journal of Applied Probability}, vol.~24, no.~3, pp.~600--608, 1987.

\bibitem{wolf1980approximation}
D.~Wolf, ``Approximation of the invariant probability measure of an infinite stochastic matrix,'' {\em Advances in Applied Probability}, vol.~12, no.~3, pp.~710--726, 1980.

\bibitem{tweedie1998truncation}
R.~Tweedie, ``Truncation approximations of invariant measures for markov chains,'' {\em Journal of applied probability}, vol.~35, no.~3, pp.~517--536, 1998.

\bibitem{infanger2022convergence}
A.~Infanger and P.~W. Glynn, ``On convergence of a truncation scheme for approximating stationary distributions of continuous state space markov chains and processes,'' {\em arXiv preprint arXiv:2206.11738}, 2022.

\bibitem{seneta1980computing}
E.~Seneta, ``Computing the stationary distribution for infinite markov chains,'' {\em Linear Algebra and Its Applications}, vol.~34, pp.~259--267, 1980.

\bibitem{meyn2012markov}
S.~P. Meyn and R.~L. Tweedie, {\em Markov chains and stochastic stability}.
\newblock Springer Science \& Business Media, 2012.

\bibitem{wellner2013weak}
J.~Wellner {\em et~al.}, {\em Weak convergence and empirical processes: with applications to statistics}.
\newblock Springer Science \& Business Media, 2013.

\bibitem{wang2015construction}
R.~Wang, J.~Honda, H.~Yamamoto, R.~Liu, and Y.~Hou, ``Construction of polar codes for channels with memory,'' in {\em 2015 IEEE Information Theory Workshop-Fall (ITW)}, pp.~187--191, IEEE, 2015.

\end{thebibliography}
\bibliographystyle{ieeetr}

\end{document}